%% file: main.tex
\documentclass[runningheads]{llncs}
\usepackage[T1]{fontenc}
\usepackage{graphicx}
\usepackage{custom-macros}

\usepackage{color}

\begin{document}
\title{Runtime Monitoring of Distributed Cyber-Physical Systems Without a Global Clock\thanks{This research was partially supported by NSF awards 2118179, 2145291 and 2416461.}}

\author{Charles Koll\orcidID{0000-0001-5941-250X} \and
Houssam Abbas\orcidID{0000-0002-8096-2618}}
\authorrunning{C. Koll, H. Abbas}
\institute{Oregon State University, Corvallis OR 97331, USA\\
\email{\{kollch,abbasho\}@oregonstate.edu}}
\maketitle              % typeset the header of the contribution
\begin{abstract}
We give the first theoretical characterization, and the first algorithm, for continuous monitoring of a distributed Cyber-Physical System (CPS) against a dense-time temporal logic specification. A distributed CPS is composed of multiple agents, each with a local clock; these clocks drift from each other, so there is no well-defined global time. When monitoring such a system's output signal against a temporal logic specification, it is not evident how to interpret the temporal constraints of the formula, and what satisfaction means. Yet CPS designers, like control engineers, typically think of their system's operation in terms of global time. Most existing techniques for monitoring distributed systems work with discrete-time specifications not suitable for CPS, and/or require an explicit mapping of temporal constraints to local clocks. We introduce an algorithm that addresses the above challenges for a fragment of Signal Temporal Logic (STL) that still includes all temporal operators. It relies on a novel extension of satisfaction signals to this partially synchronous setting (where clocks drift), and an analysis of the geometry of multi-dimensional partially synchronous time. The algorithm returns the set of all possible global moments that can satisfy the specification. Knowledge of these possible global moments is important for debugging distributed hybrid control systems such as fleets of drones and electrical grids. We derive the worst-case complexity of the algorithm, and implement a sound approximation of it that experimentally illustrates effective monitoring, even in scenarios of up to 50 agents.

\end{abstract}
\section{Introduction: Monitoring Dense-Time Requirements Without a Shared Clock}
\label{sec:intro}
Cyber-Physical Systems (CPS) are key components of the devices and processes around us. 
A {\em distributed CPS} consists of multiple communicating processes or {\em agents}. Examples include networks of autonomous vehicles \cite{feng2018design}, fleets of drones \cite{wang2019survey}, smart controllers in electrical grids \cite{gavriluta2020cyber}, and geographically dispersed sensor networks \cite{chen2015ubiquitous}.
A common formalism for CPS is a hybrid system, whose output signals can evolve both continuously and discontinuously.

Regardless of the distributed CPS' design process, it is often necessary to monitor its operation at runtime, to detect whether it is breaching the CPS' correctness requirements. 
These requirements are commonly expressed as a formula in some temporal logic.
Three challenges face monitoring temporal logic requirements in distributed CPS, compared to monitoring traditional discrete distributed systems:
first, the individual agents that make up the system each have a local clock, and these clocks drift from each other. When two agents report a value of their signal at local time $t$, these two values are not necessarily synchronous. If the temporal logic formula says `Within 5 sec, $p$ is true', the monitor must find a reasonable interpretation of the temporal constraint `Within 5 secs' -- for example, on which local clock are 5 seconds to be measured? Is it reasonable to use just one local clock?
Second, as CPS include a physical aspect (e.g. a car's dynamics or the physical organ controlled by a medical device), they measure and compute over signals in \textit{dense time}, aka continuous time; for any finite time span, there are uncountably many `events' in the system.
Most work in distributed systems monitoring deals with discrete (or logical) time with at most countably many events (e.g. \cite{fabre2002monitoring,tekken2017monitoring,zhao2001distributed,ganguly2022distributed}).
Third, monitoring is usually (though not always) a continual process: we want to determine whether the system satisfies the specification at every moment, not only at global time 0. So we cannot rely on a monitor that only returns a punctual verdict (one that holds at a given time $t$), since there are uncountably many moments in dense time.

This work addresses these challenges for \textit{partially synchronous} systems: such systems use an algorithm, like NTP \cite{ntp}, to keep their clocks within a known bound $\cskew$ of each other.
Signal values occurring within $\cskew$ time units of each other might be synchronous, and so should all be explored by the monitor for possible violation or satisfaction.
This paper develops a theoretical characterization of the geometry and dynamics of the set of possibly synchronous satisfying moments, aka \emph{satcuts}.
On that basis it develops an algorithm that approximates, to arbitrary precision, the set of \textit{all} satcuts. 
A fragment of Signal Temporal Logic (STL) \cite{maler2004monitoring} is used as specification language. STL is widely used for specifying requirements of CPS, such as `At every moment between 0 and $100$ ms, a critical separation is followed, 6 to 10.5 ms later, by a negative acceleration'.
By using plain STL, the CPS designer can continue to treat the system as though it were perfectly synchronous, leaving the burden of dealing with asynchronicity to the monitor. 
We give an offline monitor for a significant fragment of STL that we name \emph{\distl}, characterize its complexity, and develop an implementation for experimental evaluation.

\emph{Related work.}
There is a vast literature on distributed digital systems, which can be surveyed in \cite{garg02book}. 
Most of this work uses discrete or logical time, such as \cite{fabre2002monitoring,tekken2017monitoring,zhao2001distributed,ganguly2022distributed}, and is not applicable here. 
A result from \cite{chase1998detection} shows that the complexity of monitoring such a system in general is NP-complete.
Two papers \cite{momtaz2023predicate,koll2023decentralized} address monitoring dense-time distributed systems, but only for the special case of boolean predicates and not temporal logic specifications.

Existing temporal logics for specifying properties of distributed systems, like \cite{basin2011distributed,baumeister2021temporal,sen2004efficient}, do not preserve the abstraction of a single synchronized system for the engineer designing the CPS. 
STL \textit{temporal robustness} is introduced in \cite{donze2010robust} to quantify by how much a signal can be shifted in time while preserving its truth value relative to an STL formula. This is a special case of our setting: in our partially synchronous setting, signals can shift by varying amounts at different points in time. Thus we explore a much broader set of synchronizations than is measured by temporal robustness.
This difference of constant versus variable temporal shifts of the signal prevents us from using their approach for monitoring a distributed system.

Finally, \cite{momtaz2023monitoring} does online STL monitoring using an SMT solver, but only returns whether the spec is satisfied at time 0, while we return all such (possibly) synchronous moments.
In addition, we are able to characterize the complexity of our algorithm in terms of meaningful quantities, like the quality of our approximations, rather than in terms of the number of variables in a particular SMT encoding.
Thus there is no work directly comparable to this work.

{\em Contributions.} This paper develops the theory of partially synchronous monitoring in dense time. We:
\begin{enumerate}
    \item Identify a fragment of STL, \emph{\distl}, that is amenable to monitoring over partially synchronous systems. This fragment includes all temporal operators.
    \item We characterize the geometry of the set of satcuts, which are possibly synchronous moments, from the $N$ agents' timelines, that satisfy the STL formula.
    \item We provide the first offline monitor that approximates, to arbitrary precision, the entire set of satcuts (aka the satisfaction signal).
    \item We also provide an implementation -- an outer approximation of the monitor that demonstrates efficient monitoring of distributed systems with 50+ agents.
\end{enumerate}

{\em Organization}. 
Preliminaries of distributed systems and STL are in \autoref{sec:prelims}, the problem and STL fragment \distl{} are formulated in \autoref{sec:problem formulation}, the monitor is described in \autoref{sec:offline mon} along with its complexity, and the implementation and experimental results are illustrated in \autoref{sec:experiments}. \autoref{sec:conclusion} concludes.

%========================================================
\section{Preliminaries}\label{sec:prelims}
The set of reals is denoted by $\reals$, the set of non-negative reals by $\nnreals$, and the set of positive reals as $\preals$. The set $\{1, 2, \dots, N\}$ is abbreviated as $[N]$. Given a set $S$, $2^S$ is the set of subsets of $S$. 

In $\reals^N$, a {\em polytope} is a bounded intersection of half-spaces, and a {\em non-convex polytope} is a connected union of polytopes: so between any two points in the set it is possible to draw a curve, \textit{not necessarily a straight line}, that belongs entirely to the set. A {\em box} $B$ is an axis-aligned polytope, i.e. $B = [a_1,b_1]\times\ldots \times [a_N,b_N]$ for some reals $a_i,b_i$. 
A \emph{lower boundary} of $B$ is any hyperplane $\{x\in \reals^N \such x_i=a_i\}$.
Given two subsets $S_1$ and $S_2$ of $\nnreals^N$, we define $S_1 \oplus S_2 \defeq \{s_1 + s_2 \such s_1 \in S_1, s_2 \in S_2\}$ and $S_1 \ominus S_2 \defeq  \nnreals^N\cap \{s_1 - s_2 \such s_1 \in S_1, s_2 \in S_2\}$. For $a \in \nnreals^N$, we write $a \oplus B$ and $a \ominus B$ for $\{a\} \oplus B$ and $\{a\} \ominus B$, respectively.

Reference (hypothetical) time values are denoted by $\gclk$, $\gclk'$, etc, while $t$, $t'$, $t_1$, $t_2$, $s$, $s'$, $s_1$, $s_2$, etc. denote \emph{local} clock values specific to given agents.
Given a vector $v$ its $k^{th}$ element is written $v[k]$.

Various proofs in this work are sketched. The remaining (full) proofs can be found in the appendix.

\subsection{Signal Model}
We consider a system consisting of \emph{$N$ agents} that do not fail, denoted by $\{\agent_1, \dots, \agent_N\}$, without any shared memory or global clock. The output signal of agent $\agent_n$ is denoted by $x_n$, for $n \in [N]$. It is simply a function of time.
A \emph{right-continuous} signal is one s.t. at all $t$ in its support, $\lim_{s \rightarrow t_+} x(s) = x(t)$. It is \emph{left-limited} if it has a finite left-limit at every $t$ in its support: $\lim_{s \rightarrow t_-} x(s) < \infty$. A \emph{Zeno} signal has an infinite number of discontinuities in at least one bounded interval in its support. A discontinuity in signal $x(\cdot)$ can be due to a discrete event in the agent (like a variable update by software).

\begin{definition}[Output/Distributed signals]
    \label{def:dstr signal}
    Let $\timeline \defeq \nnreals$. An \emph{output signal} of some agent $\agent$ is a function $x: \timeline \mapsto \reals^d$, which is right-continuous, left-limited, and is not Zeno. We refer to $\timeline$ as the \emph{timeline} for agent $\agent$.

    A \emph{distributed signal} $\dstrsig \defeq (x_n)_{n \in [N]}$ on $N$ agents is a collection of $N$ output signals.
\end{definition}

Without loss of generality, we assume that $x$ is one-dimensional, i.e., $d = 1$. 

Intuitively, each local clock increases strictly (as real time passes), and all local clocks remain within a known bound from each other, e.g. by using a synchronization algorithm like NTP \cite{mills2010network}.
To model these properties we will need to refer to a hypothetical reference clock $\gclk$. This reference clock is a purely hypothetical object used in definitions and theorems, and is not a real clock nor is it used in the algorithms.

\begin{assumption}[Partial synchrony]\label{ass:partialsync}
    The \emph{local clock} of an agent $\agent_n$ can be represented as a strictly increasing function $c_n: \timeline \mapsto \timeline$, where $c_n(\gclk)$ is the value of the local clock at reference time $\gclk$. 
    Moreover there exists a positive real $\cskew$ s.t. for any agent $\agent_n$ and for all $\gclk$ in $\timeline$, $| c_n(\gclk) - \gclk | \leq \cskew$. The constant $\cskew$ is the maximum \emph{clock skew}, which is assumed fixed and known by the system designer. Finally, for all $n$, $c_n(0)=0$.
\end{assumption}

The maximum clock skew $\cskew$ is fixed, but signals can drift arbitrarily within $\cskew$. This enables us to model transient phenomena of NTP and other such algorithms. 
An {\em event on agent $n$} is simply a local time/value pair, $(t,x_n(t))$, though we will often abuse notation and talk of event $t$, when the agent and output signal value are understood or irrelevant.
A set of $N$ events $t_1,\ldots,t_N$, one per agent, is \emph{concurrent} if $|t_n-t_m|\leq 2\cskew$ for all $n,m \in [N]$. Thus, given the maximum skew, it is not possible to tell the temporal ordering of events inside a concurrent set, and our monitors must treat them as being possibly synchronous. 

The notion of consistent cut was defined for partially synchronous dense time signals in \cite{momtaz2023predicate} by extending the classical notion of consistent cut from distributed systems. It is easy to show that the following is an equivalent definition, which has the benefit of greater simplicity when describing our algorithms.
\begin{definition}[Consistent cut]
    Let $\dstrsig$ be a distributed signal over $N$ agents. 
    A {\emph cut} $C$ is a set of local timestamps of the form: $C =[0,t_1]\times [0,t_2]\times\ldots\times[0,t_N]$,
    where $t_n$ is measured on $\agent_n$'s local clock.
    The {\em frontier} $\front(C)$ of a cut $C$ is its vector of final timestamps: $$\front(C) \defeq (t_1,\ldots,t_N)\in \nnreals^N.$$
    A \emph{consistent cut}, or concut, is a cut whose frontier is concurrent, i.e. $|t_n-t_m|\leq 2\cskew$ for all $t_n,t_m$ elements of $\front(C)$.

    Finally, given concuts $C=[0,t_1] \times \ldots \times [0,t_N]$ and $C'=[0,t_1'] \times \ldots \times [0,t'_N]$, the partial order $\prec$ between concuts is given by $C \prec C'$ iff $t_n < t_n'$ for all $n$.
The non-strict version is $C \preceq C'$ iff either $C \prec C'$ or $C = C'$.
\end{definition}

We write $\boldsymbol{0}$ for the concut whose frontier is the 0 time at each agent. 
Note that for two different concuts to be ordered ($C \prec C'$), \textit{all} clocks need to have advanced from $C$ to $C'$: we don't allow some local clocks to progress while others stall.
We write $[C,C']$ for the set of cuts $C''$ s.t. $C \preceq C'' \preceq C'$. 

\subsection{Signal Temporal Logic (STL)}\label{sec:stl}
The system designer formalizes requirements in {\em Signal Temporal Logic}, or STL \cite{maler2004monitoring}, which is a common logic for CPS formal requirements.
Let \textsf{AP} be a set of \emph{atomic propositions}. To every $p\in AP$ is associated an $N$-ary function $f_p: \reals^N \mapsto \reals$. The syntax for STL is given by:
$$\stlfm{\varphi} \defeq \top \mid p \mid \lnot\stlfm{\varphi} \mid \stlfm{\varphi} \land \stlfm{\varphi} \mid \stlfm{\varphi} \until{[a,b]} \stlfm{\varphi}$$

STL combines the usual boolean constant True $\top$ and boolean operators (negation $\neg$ and conjunction $\land$) with the temporal Until operator $\varphi \until{[a,b]} \psi$, which means that $\varphi$ remains true at least until a moment in $[a,b]$ at which $\psi$ becomes true. 
Formally, let a \emph{trace} $\sigma \defeq (x_1, \dots, x_N)$ be a vector of $N$ continuous-time \textit{synchronous} output signals, one per agent. We think of $\sigma$ as a regular $N$-dimensional output signal whose values are measured against global time. 
The satisfaction of formula $\varphi$ by trace $\sigma$ at time $t$, written $(\sigma,t)\models \varphi$, is defined by:

{\centering
\begin{tabular}{ l c l }
    $(\sigma, t) \models \top$ & & \\
    $(\sigma, t) \models p$ & iff & $f_p(x_1(t), \dots, x_N(t)) \geq 0$ \\
    $(\sigma, t) \models \lnot\stlfm{\varphi}$ & iff & $(\sigma, t) \not\models \stlfm{\varphi}$ \\
    $(\sigma, t) \models \stlfm{\varphi} \land \stlfm{\psi}$ & iff & $(\sigma, t) \models \stlfm{\varphi}$ and $(\sigma, t) \models \stlfm{\psi}$ \\
    $(\sigma, t) \models \stlfm{\varphi} \until{[a,b]} \stlfm{\psi}$ & iff & $\exists t' \in [t + a, t + b]: (\sigma, t') \models \stlfm{\psi}$ and \\
    &&$\forall t'' \in [t, t']: (\sigma, t'') \models \stlfm{\varphi}$
\end{tabular}
\par}

For convenience we write $\sigma \models \stlfm{\varphi}$ for $(\sigma, 0) \models \stlfm{\varphi}$. Two additional temporal operators can be defined with $\until{[a,b]}$, $\lnot$, and $\top$: Eventually ($\eventually_{[a,b]} \varphi \defeq \top \until{[a, b]} \varphi$) and Always ($\always_{[a,b]} \varphi \defeq \lnot \eventually_{[a,b]} \lnot \varphi$).

%======================================================================================================
\section{Problem Formulation}
\label{sec:problem formulation}
As stated in the Introduction, we wish to allow the designer to write specifications in STL, without burdening them with the need to account for partial synchrony in the logic itself, e.g., by choosing on which local clock an interval is to be evaluated: this choice would be both arbitrary (why one clock and not another?) and wrong (the engineer means for the interval to measure global time, not local time). 
So we must first define what it means for a distributed signal to satisfy an STL formula. 
\textit{Retimings} \cite{momtaz2023predicate} are a theoretical construct that maps local clock values to a reference time. 

\begin{definition}[Retimings and distributed STL satisfaction]
\label{def:retimings and distributed stl sat}
    Fix $\cskew >0$.
    An \emph{agent retiming} for agent $n$ is a curve $\rho_n: \nnreals \mapsto \timeline$ s.t. $\rho_n(0) = 0$, $\rho_n$ is strictly increasing, and for all $t$ in $\nnreals$, $| t - \rho_n(t) | \leq \cskew$.

    A \emph{system retiming}, or simply \emph{retiming}, is a curve $\rho: \nnreals \mapsto \timeline^N$ where
    $$\rho(\omega) \defeq (\rho_1(\omega), \dots, \rho_N(\omega)).$$
    
    For readability we say the retiming $\rho$ uses agent retimings $\rho_1, \dots, \rho_N$ for its definition, retiming $\rho'$ uses agent retimings $\rho_1', \dots, \rho_N'$, retiming $\rho''$ uses $\rho_1'', \dots, \rho_N''$, etc. unless explicitly stated otherwise. The \emph{graph} of $\rho$ is $Gph(\rho) \defeq \{(\rho_1(\omega), \dots, \rho_N(\omega), \omega) \such \omega \in \timeline\}$; we place $\omega$ at the end of the tuple for simplicity when we connect retimings and concuts in a later section. Given a distributed signal $\dstrsig \defeq (x_n)_{n \in [N]}$, the \emph{synchronized signal} $\dstrsig_\rho$ is the trace $\sigma \defeq (x_1\circ \rho_1,\dots, x_N \circ \rho_N)$, where $\circ$ is the standard function composition operator.
    
    We say that $\dstrsig$ satisfies STL formula $\varphi$ at time $\omega$ iff there exists a retiming $\rho$ s.t. $ (\dstrsig_\rho,\omega) \models \varphi$.
\end{definition}

Similar retiming definitions have been described in prior works such as \cite{koll2023decentralized}, \cite{momtaz2023predicate}, and \cite{quesel2011crossing}.

\subsection{Why the STL Fragment?}
\label{sec:why distl}
Monitoring the full STL logic raises serious difficulties that we now detail. This justifies restricting attention to a fragment of STL, \emph{\distl}, defined in the next section.
Consider $\varphi \defeq \psi_1 \land \psi_2$ for general STL formulas $\psi_1,\psi_2$. Fix a time $\omega$.
Suppose that there exist retimings $\rho'$ and $\rho''$ with $\omega \in \nnreals$ s.t. $(\dstrsig_{\rho'}, \omega) \models \psi_1$ and $(\dstrsig_{\rho''}, \omega) \models \psi_2$. 
These retimings may not be the same. 
But to say that $(\dstrsig,\omega) \models \psi_1\land \psi_2$, there must exist the \emph{same} retiming $\rho$ where $(\dstrsig_\rho, \omega) \models \psi_1$ and $(\dstrsig_\rho, \omega) \models \psi_2$.
Thus we would need to finitely represent and keep track of {all} retimings that witness satisfaction in the subformulas to evaluate satisfaction of the overall formula.
This presents two challenges: first, creating a finite representation of retimings necessarily means restricting our retimings to some class of functions, e.g. piecewise linear increasing functions; and second, keeping track of \textit{all} possibly satisfying retimings across sub-formulas can very quickly become exorbitantly expensive, both computationally and memory-wise. 
\distl{} restricts the language to a fragment that does not require tracking of retimings, allowing us to sidestep these concerns.

While \distl{} allows for efficient monitoring in the presence of retimings, we do note a limitation of this fragment: preventing generalized conjunctions ($\psi_1 \land \psi_2$ for general formulas $\psi_1$ and $\psi_2$) prevents a formula designer from introducing a constraint to the formula $\psi_1$ by merely conjuncting an additional arbitrary formula $\psi_2$ to form $\psi_1 \land \psi_2$ -- in \distl, the formula $\psi_2$ is limited in what it can specify.

\subsection{The DiSTL Fragment} \label{sec:distl}
Our next move is to define the fragment \distl. 
A general formula of this fragment will be labeled $\stlsubfm$.

\begin{definition}[DiSTL]\label{def:stl gram} DiSTL is a fragment of STL with syntax:
    $$\stlsubfmatom \defeq \top \mid \bot \mid p_n \mid \lnot p_n \mid \stlsubfmatom \lor \stlsubfmatom \mid \stlsubfmatom \land \stlsubfmatom$$
    $$\stlsubfm \defeq \stlsubfmatom \mid \stlsubfm \lor \stlsubfm \mid \stlsubfm \land \stlsubfmatom \mid \eventually_{I} \stlsubfm \mid \always_{I} \stlsubfmatom \mid \stlsubfmatom \until{I} \stlsubfm$$
    where $I = [a,b]$ is an interval such as that attached to the Until in the STL syntax (\autoref{sec:stl}).

    In the semantics, $f_{p_n}(x_1(t),\dots,x_N(t)) = x_n(t) - \beta_n$ for some rational number $\beta_n$ (see \autoref{sec:stl}). 
\end{definition}
For monitoring purposes satisfaction must be defined in terms of what can be observed, namely, cuts. So how do we go from satisfaction in terms of retimings (\autoref{def:retimings and distributed stl sat}) to satisfaction in terms of cuts?
We make three critical observations:
\begin{enumerate}
    \item The first observation is that for a synchronous trace $\sigma$, the index of evaluation is a single (global) moment $\gclk$. 
For a distributed signal, it must be a consistent cut's frontier, i.e. a concurrent set of local moments $(t_k)_{k \in [N]}$, since a frontier represents a \textit{potential} synchronous value of the signal.
    \item The second observation is that given two global moments $\gclk < \gclk'$, time increases in one way between them (namely it increases at a rate of 1), so it is possible to say things like `for all $\gclk''$ in $[\gclk,\gclk']$' unambiguously.
In a distributed signal, given two cuts $C$ and $C'$ with $C \prec C'$, time can evolve between their frontiers in many ways, since the local clocks can increase at different and time-varying rates. 
    \item The third observation is that because of the skew constraint on agent retimings (\autoref{def:retimings and distributed stl sat}), given a concut frontier $(t_k)_{k \in [N]}$, the tuple \\ $(t_1, \dots, t_N, \rho^{-1}(t_1,\ldots,t_N))$ is also a concut frontier but over $N+1$ agents. The last entry merely acts as (one possible) reference timestamp whose only purpose is to ensure that the local timestamps have bounded drift from each other by forcing them to be $\cskew$-away from the reference.
    Thus two different vectors $(\timeVec,\omega)$ and $(\timeVec,\omega')$ in $\nnreals^N \times \nnreals$, sharing the same first $N$ values $\timeVec$, represent two different valid retimings of the $N$ local clocks (namely, $\rho$ and $\rho'$ where $\rho(\omega) = \timeVec$ and $\rho'(\omega') = \timeVec$).
In this paper, it will be convenient for us to always treat concut frontiers as being elements of $\nnreals^{N+1}$ that obey the $\cskew$-constraint. Any concut frontier is henceforth an element of the set:
\[\Delta_\cskew^{N+1} \defeq \{v \in \nnreals^{N+1} \such \forall n: |v[n]-v[N+1]|\leq \cskew\}.\]
We connect this to retimings by saying that retimings are curves existing in this set -- for a retiming $\rho$, $Gph(\rho) \subset \Delta_{\cskew}^{N+1}$.
\end{enumerate}

\textit{Remark.}
The notation is simplified in two ways: 
(1) we will always only work with the frontier and not the whole cut. Therefore in the remainder of this work when we refer to a cut we will be referring to its frontier, notated as the same symbol $C$, which is now a vector in $\eccut$.
(2) the monitored distributed signal $\dstrsig$ is held fixed, so we drop it from the notation wherever possible, such as in the following definition.

We are now in a position to generalize satisfaction signals to the partially synchronous setting. 
In classical (synchronous) systems, the satisfaction signal of a formula is simply a function $\lambda_\varphi:\timeline \rightarrow \{\top,\bot\}$ s.t. $\lambda_\varphi(t)=\top$ iff $x,t\models \varphi$. In the current setting, there is no global time. The satisfaction signal must be defined as a function of concuts, and involve a choice of retiming.

\begin{definition}[Satisfaction Signal]\label{def:sat sig}
    Given a concut $C$, let $P_C$ be the set of retimings which pass through $C$, i.e. s.t. $C \in Gph(\rho)$. Let $\beta_n$ be a rational number.
    The {\em satisfaction signal}, or satsignal, $\lambda_{\stlsubfm}: \eccut \rightarrow \{\top,\bot\}$ is a function, parameterized by $\stlsubfm$, defined by: 
    \begin{align*}
        \lambda_\top(C) &\defeq \top,
        ~ \lambda_\bot(C) \defeq \bot,
        ~ \lambda_{p_n}(C) \defeq x_n(C[n]) \geq \beta_n,\\ 
       \lambda_{\lnot p_n}(C) &\defeq x_n(C[n]) < \beta_n, \\
        \lambda_{\stlsubfm \lor \stlfm{\psi}^\tau}(C) &\defeq \lambda_{\stlsubfm}(C) \lor \lambda_{\stlfm{\psi}^\tau}(C),\quad 
        \lambda_{\stlsubfm \land \stlfm{\psi}^a}(C) \defeq \lambda_{\stlsubfm}(C) \land \lambda_{\stlfm{\psi}^a}(C) \\
        \lambda_{\eventually_I \stlsubfm}(C) &=\top \text{ iff } \exists C' \succeq C \text{ with } C'[N+1] \in C[N+1] \oplus I: \lambda_{\stlsubfm}(C') \\
        \lambda_{\always_I \stlsubfmatom}(C) &=\top \text{ iff } \exists \rho \in P_C: \forall C' \in Gph(\rho) \text{ with }\\
        &C'[N+1] \in C[N+1] \oplus I,~ \lambda_{\stlsubfmatom}(C') \\
        \lambda_{\stlsubfmatom \until{I} \stlfm{\psi}^\tau}(C) &=\top \text{ iff } \exists C' \succeq C \text{ with } C'[N+1] \in C[N+1] \oplus I: \lambda_{\stlfm{\psi}^\tau}(C') \\
        & \text{ and } 
        \exists \rho \in P_C \cap P_{C'} \text{ s.t. } \forall C'' \in Gph(\rho) \cap [C, C']: \lambda_{\stlsubfmatom}(C'').
    \end{align*}
\end{definition}

\begin{theorem}\label{thm:valid sat sig}
    Given DisTL formula $\stlsubfm$ and concut $C \in \eccut$, let $\omega \defeq C[N+1]$. 
    Then $\dstrsig$ satisfies $\stlsubfm$ iff $\lambda_{\stlsubfm}(C) = \top$.
\end{theorem}
We can finally formulate:
\begin{mdframed}
\textbf{Main Problem}
    Given a distributed signal $\dstrsig$ and a DiSTL formula $\stlsubfm$, find \textit{all} concuts $C$ of $\dstrsig$ s.t. $\lambda_{\stlsubfm}(C)=\top$.  
\end{mdframed}

Here we note that while we aim to identify the concuts where the distributed signal \emph{satisfies} the formula, this can be considered equivalent to the CPS breaching its correctness requirements -- the correctness requirements are the negation of the provided formula. We aim to find all possible synchronizations of the system where the CPS does not follow its requirements.

%============================================================================================================
\section{Offline Monitoring: Characterization and Algorithm}
\label{sec:offline mon}
This section gives a characterization of the satdomain of \distl{} formulas, and develops an offline monitor on that basis. 
An offline monitor has access to the entire distributed signal and returns the support of the satsignal, that is, the set of concut frontiers $C$ where the satsignal evaluates to True.
We call this set of frontiers a {\em satdomain}, and denote it $\satdom{\stlsubfm}$:
\[\satdom{\stlsubfm} \defeq \{C \in\eccut \such \lambda_{\stlsubfm}(C) = \top\}\]

\subsection{Characterizing the Satdomain}
\label{sec:characterizing the satdomain}
The following lifting operator will be needed in what follows. It takes in a closed real interval $I  = [a,b]$ (which appears on a DiSTL temporal operator) and lifts it to an $(N+1)$-dimensional interval. 
        \[\ilift{I} \defeq \begin{cases}
            \preals^N \times I & \text{ if } 0 \notin I \\
            \left(\preals^N \times (I \setminus \{0\})\right) \cup {\boldsymbol{0}} & \text{ otherwise.}  
        \end{cases}\]

This lift is meant to capture what should happen to local clocks if the formula contains a temporal operator decorated by interval $I$.
Recall that $I$, as in $\eventually_I p$, requires that time must shift by some $c \in I$ before $p$ is usefully True.
If $I$ does not start at 0 then the $N$ local clocks, whose values appear in the first $N$ dimensions of $\ilift{I}$, must shift by a non-0 amount (as required by \autoref{ass:partialsync}), thus we cross $I$ with $\preals^N$.
If $I=[0,a]$, the possibility that the shift is 0 implies all other local clocks have not moved either, thus we make sure the all-0 vector is in the lifted $N+1$-dimensional interval. On the other hand the possibility that the shift is positive implies $I\setminus \{0\}$ is crossed, as before, with $\preals^N$.

The following theorem gives a recursive way to compute the satdomain of a DiSTL formula, excluding the Until case.

\begin{theorem}\label{thm:satdomain captures satisfaction}
\label{thm:satdomain equalities}
    The following equalities hold: 
    \begin{align*}
        \satdom{\top} &= \eccut~,~\satdom{\bot} = \emptyset,~\quad~
        \satdom{p_n} = \{C \in \satdom{\top} \mid x_n(C[n]) \geq \beta_n\}
        \\
        \satdom{\stlsubfm \lor \stlfm{\psi}^\tau} &= \satdom{\stlsubfm} \cup \satdom{\stlfm{\psi}^\tau},~ \quad~\satdom{\stlsubfm \land \stlfm{\psi}^a} = \satdom{\stlsubfm} \cap \satdom{\stlfm{\psi}^a} \\
        \satdom{\eventually_I \stlsubfm} &= \satdom{\stlsubfm} \ominus \ilift{I},
        \\
        \satdom{\always_{[b,c]} \stlsubfmatom} &= \satdom{\eventually_{[b,b]}\left(\stlsubfmatom \until{[c-b,c-b]} \top\right)}
    \end{align*}
\end{theorem}

Now the Until is treated.
The set $\satdom{\stlsubfm\until{I}\psi^\tau}$ contains all concuts $C$ that 
\begin{enumerate}[label=(U\arabic*)]
    \item are in $\satdom{\stlsubfm}$,
    \item at which starts a strictly increasing curve $\gamma: [0,1] \rightarrow \eccut$ which ends at some $C'$ in $\satdom{\psi^\tau}$ s.t.
    \item this curve lives entirely in $\satdom{\stlsubfm}$ at least until it reaches $C'$ and
    \item the difference between $C[N+1]$ and $C'[N+1]$ is in $I$. 
\end{enumerate}
The curve is another representation of a retiming: by increasing from its start at $\gamma(0)=C$ to its end at $\gamma(1)=C'$ and living entirely in $\satdom{\stlsubfm}$ in-between, it witnesses that it is possible for the concurrency to resolve itself in such a way that $\stlsubfm$ is satisfied until $\psi^\tau$ is satisfied. 
This characterization of the satdomain of an Until operator will be used in the following section, in which we introduce the offline \distl{} monitor for partially synchronous distributed signals.

\subsection{The Offline Monitor}
\label{sec:monitor without until}
We now introduce our offline monitor, which takes in a \distl{} formula $\stlsubfm$, a distributed signal $\dstrsig$, and recursively computes the satdomain $\satdom{\stlsubfm}$.

For the base cases of $\top$ and an atom $p_n$, and for all operators except the Until, the monitor simply performs the operations indicated in \autoref{thm:satdomain captures satisfaction}. These are straightforward polytope manipulations.
 
\begin{lemma}
    \label{thm:satdomains are NCPs}
    The sets $\satdom{\top}$, $\satdom{p_n}$ and $\satdom{\neg p_n}$ are non-convex polytopes (NCPs). 
    Further, if $\satdom{\stlsubfm}$, $\satdom{\stlfm{\varphi}^a}$ and $\satdom{\stlfm{\psi}^\tau}$ are NCPs, then so are the sets $\satdom{\stlsubfm \lor \stlfm{\psi}^\tau}$, $\satdom{\stlsubfm \land \stlfm{\varphi}^a}$, $\satdom{\eventually_I \stlsubfm}$ and $\satdom{\always_{I} \stlsubfmatom}$.
\end{lemma}
\begin{proof}[Proof sketch]
The base cases are trivial. All other operations (namely, unions, intersections, and shifts) produce non-convex polytopes from non-convex polytope inputs.
\end{proof}

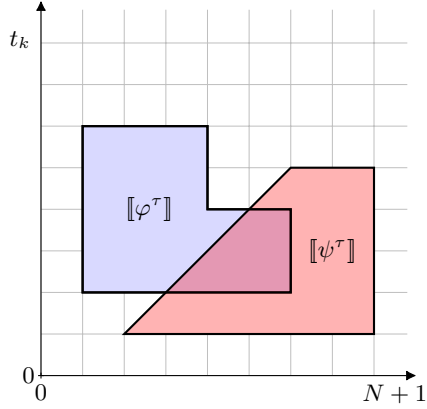
\begin{figure}[ht]
    \centering
    \input{figs/or_example}
    \caption{Example of computing $\satdom{\stlsubfm \lor \stlfm{\psi}^\tau}$ in the space of cuts. The $N+1$ dimension is shown on the x-axis, and the y-axis represents the other dimensions.}
    \label{fig:or example}
\end{figure}

\begin{figure}[!ht]
    \centering
    \input{figs/eventually_example}
    \caption{The satdomain $\satdom{\eventually_{[1,2]} \stlsubfm}$ in the space of cuts. The diagonal bands delimit $\eccut$. The operation $\satdom{\stlsubfm} \ominus (\preals \times [1,2])$ produces the set in the right diagram. Boundaries drawn with dashed lines are not included in the set. This is due to the use of $\preals$ in the operation.}
    \label{fig:eventually example}
\end{figure}
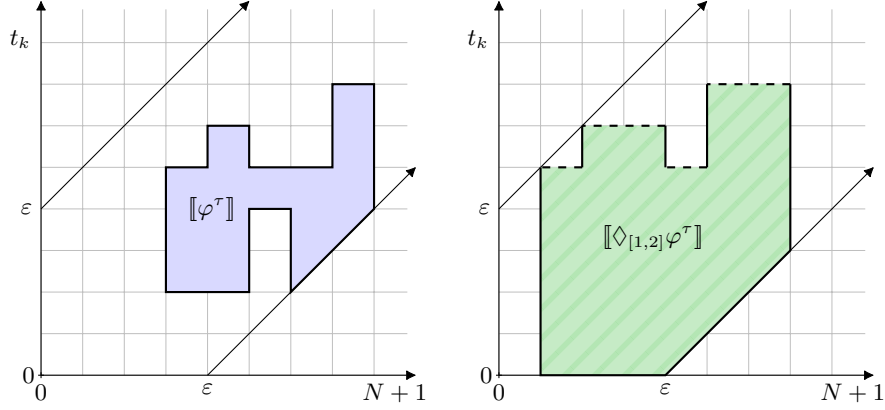

\begin{example}\label{ex:or}
\autoref{fig:or example} shows the construction of $\satdom{\stlsubfmatom \lor \psi^a}$, which is simply the union of the two constituent satdomains.
\autoref{fig:eventually example} illustrates deriving $\satdom{\eventually_{[1,2]} \stlsubfm}$ from $\satdom{\stlsubfm}$. The $N+1^{st}$ dimension is shown along the $x$-axis. 
For every cut $C$ in the left diagram, there exists a $C'$ in the right diagram which is offset by the interval $[1,2]$ and is strictly greater than $C$. 
\end{example}

\paragraph{The Until operator}
To process an Until operator, we introduce two algorithms, \moveBackIn{} and \moveBackOut.
These produce an inner and an outer approximation, respectively, of $\satdom{\stlsubfmatom \until{I} \stlfm{\psi}^\tau}$, given the input domains $\satdom{\stlsubfmatom}$ and $\satdom{\stlfm{\psi}^\tau}$.

Algorithm \moveBackOut{} is given in \autoref{algo:approx until satcuts}. 
Intuitively, it partitions the polytopes of $\satdom{\stlsubfm}$ and $\satdom{\psi^\tau}$ into axis-aligned boxes s.t. it is trivial to move from one box to a neighboring box that is strictly above it (to track time increase). The main cost is due to computing these partitions and refining them.
Because the input domains' polytopes cannot necessarily be partitioned exactly into boxes, \moveBackOut{} over-approximates them with bounding boxes where needed. This is the source of the over-approximation.
An example of its operation is given in \autoref{fig:until}.

\moveBackIn{} works just like \moveBackOut{} but with one modification: on lines 2--7, rather than find bounding boxes, it packs the non-convex polytopes with inscribed boxes of some size, thus yielding an inner approximation.

\begin{figure*}[!ht]
    \centering
    \input{figs/until_example}
    \caption{An example of  \moveBackOut$(\stlsubfmatom \until{[2,4]} \stlfm{\psi}^\tau)$. The input $\satdom{\stlsubfmatom}$ consists of the  non-convex polytope in light blue and the input $\satdom{\stlfm{\psi}^\tau}$ is the salmon non-convex polytope.  \moveBackOut{} computes the bounding boxes for each convex polytope in $\satdom{\stlsubfmatom} \cup \satdom{\stlfm{\psi}^\tau}$; these boxes are shown by thick dashed lines (\autoref{algo:approx until satcuts} lines 2--7). Next, the bounding boxes are partitioned from each vertex, shown with dashed lines (line 6). The partitioned bounding boxes for $\satdom{\stlsubfmatom}$ are then refined by intersections with the partitioned bounding boxes for $\satdom{\stlfm{\psi}^\tau}$, shown with dotted lines (lines 10--12). We identify the partitioned blocks in the intersection between the two sets, in this figure a single block $S'$ in dashed purple (line 10). The set $G$ is determined by offsetting $S'$ horizontally by the interval $[2,4]$ and extending downward to produce $G$ (shown in gray, line 11). We then move left and down the blocks from $S'$, pushing those intersecting with $G$ into a results set $Y$ (lines 15--20). For example, the block $S_1$ is added to $Y$ as it is left and down from $S'$ and intersects $G$. The algorithm returns $Y$ (hashed green).}
    \label{fig:until}
\end{figure*}
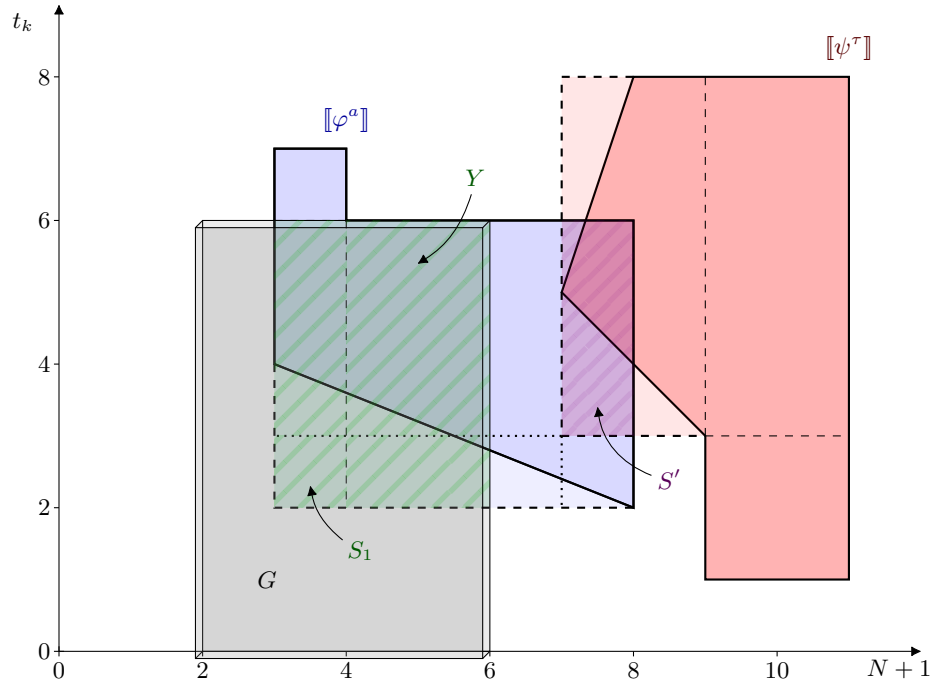

\input{algo-until} %****************************************

\begin{lemma}\label{lem:moveback approx}
    If the inputs $\satdom{\stlsubfmatom}$ and $\satdom{\stlfm{\psi}^\tau}$ to \moveBack$_{out/in}$ are non-convex polytopes, then so is the output set.
    Moreover, \moveBackIn{} and \moveBackOut{} give inner and outer approximations, respectively, of the satdomain. Formally:
    $$\text{\moveBackIn}(\stlsubfmatom \until{I} \stlfm{\psi}^\tau) \subseteq \satdom{\stlsubfmatom \until{I} \stlfm{\psi}^\tau} \subseteq \text{\moveBackOut}(\stlsubfmatom \until{I} \stlfm{\psi}^\tau).$$
\end{lemma}

\begin{proof}[Proof sketch]
    The first part is easy to establish since intersections, unions and separations of polytopes along hyperplanes produce polytopes.

    The proof for outer-approximation is done via induction. Every operator besides the Until clearly meets the lemma, by \autoref{thm:satdomain captures satisfaction}. For the Until, we analyze \autoref{algo:approx until satcuts} (\moveBackOut).  Without loss of generality, we consider a case where there is a single $S' \subseteq \overline{B}$, where $\overline{B}$ is the  outer approximation of $\satdom{\stlfm{\psi}^\tau}$ (line 10).

    To begin with, we say that a point $C'$ is $I$-away from a point $C$ when $C'[N+1] \in C[N+1] \oplus I$ and $C \preceq C'$. We see that $G$ contains all points that are $I$-away from some point in the destination $S'$. Next, we can say that for all blocks $S_1, \dots, S_k$ from $L$ (for some natural number $k$), $G \cap (S_1 \cup \dots \cup S_k)$ contains every point that (1) is $I$-away from $S'$ and (2) has a strictly increasing path to it within $\overline{A}$.
    These two remarks imply that every satcut $C$ is in some $G \cap S$, and therefore every satcut $C$ is returned by the algorithm. This concludes the proof sketch.
\end{proof}

Both versions of \moveBack{} can be tuned: the box outer approximations of \moveBackOut{} can be made arbitrarily tight by outer-approximating a given polytope with several boxes instead of one, and the inner approximations of \moveBackIn{} can be made arbitrarily tight by packing more boxes of smaller size. Thus approximation quality is traded-off against computational cost.

\paragraph{Complexity of \autoref{algo:approx until satcuts}.}
\label{sec:complexity}
Let $P_1$ and $P_2$ be the numbers of polytopes making up the non-convex polytopes $\satdom{\stlsubfmatom}$ and $\satdom{\stlfm{\psi}^\tau}$, 
$M_1$ and $M_2$ (upper bounds on) the numbers of bounding half-spaces in $\satdom{\stlsubfmatom}$ and $\satdom{\stlfm{\psi}^\tau}$, and $R_1$ and $R_2$ the number of ridges in $\satdom{\stlsubfmatom}$ and $\satdom{\stlfm{\psi}^\tau}$ respectively. 
The complexity of \autoref{algo:approx until satcuts} is then
\begin{equation}
\label{eq:complexity bound}
O\left(NP_1M_1^{\lfloor\frac{N+1}{2}\rfloor} + NP_2M_2^{\lfloor\frac{N+1}{2}\rfloor} + (R_1 + R_2)^{2N+2}\right).    
\end{equation}
Appendix~\ref{sec:appendix complexity proof} sketches a proof of this bound.

\paragraph{Putting it all together.}
In effect, because of the inner/outer approximation of the Until, the monitor consists of two algorithms: the first, \monitorIn, uses \moveBackIn{} on every Until operator, and uses the resulting set at the next level of the formula. 
The second, \monitorOut, uses \moveBackOut{}.
Putting the above lemmas together yields

\begin{theorem}
Given a \distl{} formula $\stlsubfm$ and signal $\dstrsig$, 
$$\xmonitorIn(\stlsubfm) \subseteq \satdom{\stlsubfm} \subseteq \xmonitorOut(\stlsubfm).$$
\end{theorem}

We remind the reader that the inclusions can be made arbitrarily tight by using a larger number of tighter inner and outer bounding boxes in \moveBack.

\section{A Boxed Implementation and Experimental Results}
\label{sec:experiments}
Given the complexity of \autoref{algo:approx until satcuts} (see \autoref{eq:complexity bound}), an exact implementation is impractical. 
The main source of complexity is in the partitioning of arbitrary polytopes along every edge. 
To avoid this, we modify the algorithm so as to always process boxes, and not arbitrary polytopes. 
This is done by outer-approximating the satdomain by boxes at every subformula.
We refer to this as the \textit{boxed algorithm}.
Our experimental results illustrate that this approach efficiently produces satdomains even for large numbers of agents $N$ and large clock skew $\cskew$.

\paragraph{Experimental setup.}
We develop a Rust implementation of the boxed algorithm. Experiments were conducted on a machine with an AMD Ryzen 7 PRO 7840U CPU and 16 GB of LPDDR5 RAM.
We generated continuous-time signals taking values of +1 and -1, switching between them linearly every 10 ms with probability 0.1. This provided each signal with an average root rate of 10 roots/s. Each signal had a length of 6 s.

We chose three different formulas parametrized by number of agents $N$:
\\
Formula 1: $(x_1 \geq 0) \land (x_2 \geq 0) \land \dots \land (x_N \geq 0)$
\\
Formula 2: $\eventually_{[1,2]} \always_{[0,3]} ((x_1 \geq 0) \lor (x_2 \geq 0) \lor \dots \lor (x_N \geq 0))$
\\
Formula 3: $(x_1 \geq 0) \until{[0,3]} (\always_{[0,6]} ((x_2 \geq 0) \lor (x_3 \geq 0) \lor \dots \lor (x_N \geq 0)))$

We ran two comparisons: 1) evaluating runtime for different numbers of agents, and 2) evaluating runtime at different clock skews $\cskew$. For each formula, we monitored ten different distributed signals to observe the spread of runtimes.

\begin{figure}[!ht]
    \centering
    \includegraphics[width=0.9\textwidth]{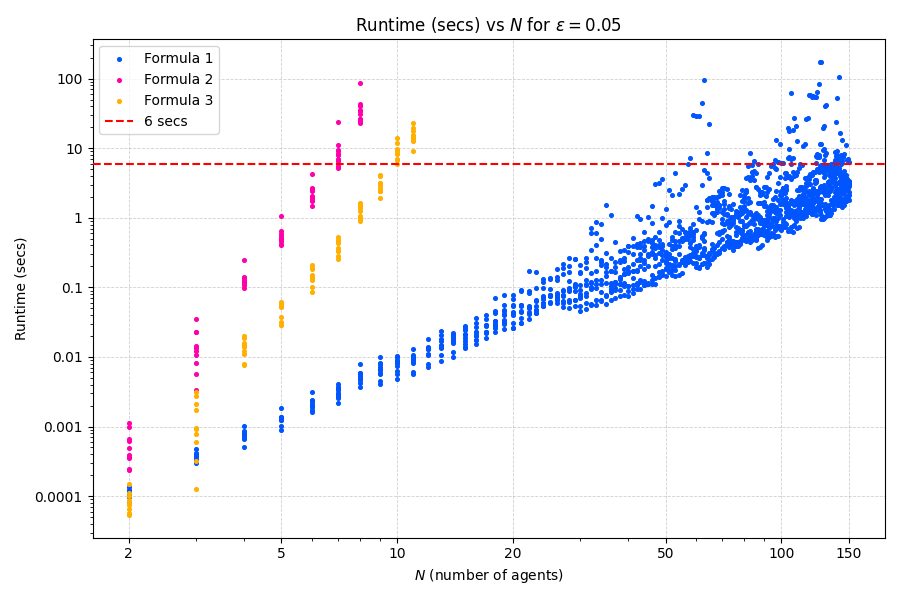}
    \caption{Runtime for varying numbers of agents (N). Ten distributed signals were run for each N and formula pairing. $\cskew = 0.05$. Plotted as log-log.}
    \label{fig:n results}
\end{figure}

\begin{figure}[!ht]
    \centering
    \includegraphics[width=0.9\textwidth]{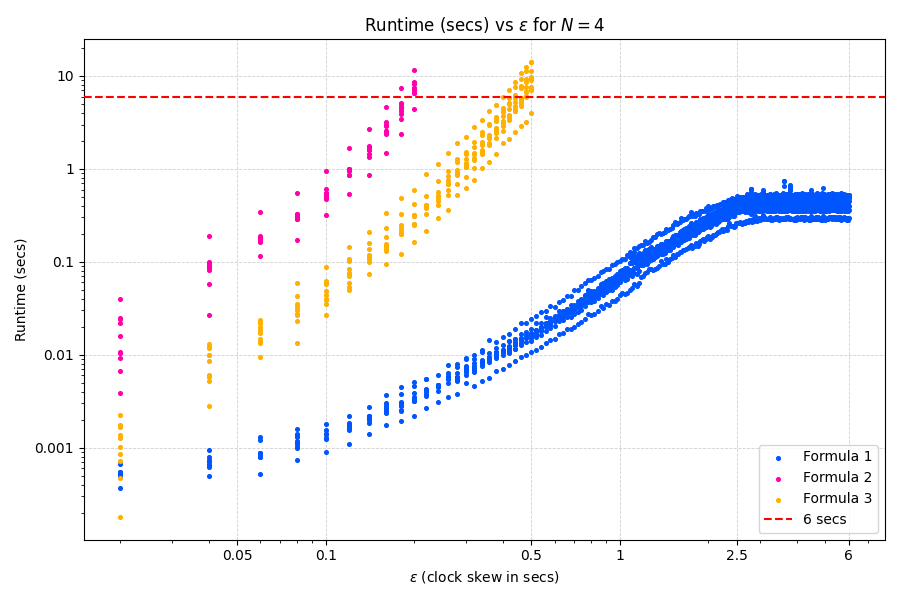}
    \caption{Runtime for varying clock skews ($\cskew$). Ten distributed signals were run for each $\cskew$ and formula pairing. $N = 4$. Plotted as log-log.}
    \label{fig:eps results}
\end{figure}

\textbf{Monitoring the boolean Formula 1 is much faster than monitoring the temporal formulas 2 and 3.} 
These runtimes are shown in \autoref{fig:n results}, which also shows the cut-off line for real-time monitoring: given that the signal is 6s long, a monitoring time of less than 6s allows real-time monitoring. 
Thus we can see that up to 57 agents can be monitored against formula 1, 9 agents against formula 3, and 6 agents against formula 2. 

\textbf{At larger $N$, monitoring formula 1 produced a wide spread in runtimes across different signals.} 
For example, at 63 agents, monitoring one signal required 95 seconds, whereas another required only 0.2 seconds. This is likely due to the number of boxes processed for each of these signals; if no intersections occur between the (satdomains of) predicates of the formula, then no further work is necessary. If instead there are many intersections, then at higher levels of the formula's syntax tree there will be more work for the intersection calculations. This illustrates the significant effect that the input signal's values have on the runtime. In general, this dependence of runtime on monitored signal holds for the classical synchronous setting as well: the more times predicates change value over the course of a signal, the more events the monitor has to track \cite{donze2010robust}.

\textbf{At fixed $N$, runtimes increase with clock skew $\cskew$, but level off for boolean formula 1.} Runtime versus clock skew is shown in \autoref{fig:eps results}. 
The increased runtime is expected as a larger skew means more events could be concurrent, so the satdomains are larger and yield more intersections.
For formula 1, runtime appears to level off around $\cskew=2.5$ seconds. This is likely because for larger $\cskew$, the boxes were not constrained by the clock skew boundary, so changing the clock skew had no effect on the operations performed. 
Formulas 2 and 3 do not see this leveling off in the experiments.

\textbf{The monitor could handle large clock skews.} Namely, the monitor is real-time (below 6 secs) for skews between 0.16 and 6 secs, depending on formula.
These are all large skews that, depending on the application domain, would likely be larger than an application's identified clock skew. Of course, as the number of agents increases, the runtimes increase as well, as explained earlier.

Overall, these results demonstrate that it is possible to do conservative real-time monitoring with our algorithm for DiSTL formulas, and that simpler formulas allow for a significant number of agents to be handled, while more complex formulas still allow for moderate group sizes at realistic clock skews. 

%================================================================================================
\section{Conclusion}
\label{sec:conclusion}
We have provided the first theoretical characterization for the satisfaction domains of dense-time temporal logic formulas in partially synchronous distributed CPS, and demonstrated the first algorithm that returns all satisfactions of such formulas, not only satisfactions at global time 0. The monitor works on a fragment of dense-time STL that includes all temporal operators. We implemented a conservative monitor based on the algorithm and demonstrated experimentally that monitoring can be effective even for large numbers of agents. Future work will address reductions in complexity, and further optimizations of the implementation.

\input{appendices}

\begin{credits}
% \subsubsection{\ackname} A bold run-in heading in small font size at the end of the paper is
% used for general acknowledgments, for example: This study was funded

\subsubsection{\discintname}
The authors have no competing interests to declare that are
relevant to the content of this article.
\end{credits}
\bibliographystyle{splncs04}
\bibliography{references}
\end{document}

%% file: figs/or_example.tex
\begin{tikzpicture}[scale=0.55]
% Axes
\draw[gray!50,very thin] (0,0) grid (8.8,8.8);
\draw[->] (0,0) -- (9,0);
\draw[->] (0,0) -- (0,9);
\draw (2pt,0) -- (-2pt,0) node[anchor=east,xshift=2pt] {\footnotesize $0$};
\draw (0,2pt) -- (0,-2pt) node[anchor=north,yshift=1pt] {\footnotesize $0$};

% Axis labels
\node[anchor=north] at (8.5,0) {$N+1$};
\node[anchor=north] at (-0.5,8.5) {$t_k$};

% Polytopes
% Polytope phi
\fill[white] (1,6) -- (1,2) -- (6,2) -- (6,4) -- (4,4) -- (4,6) -- cycle; % For hiding help lines behind phi
\fill[white] (2,1) -- (8,1) -- (8,5) -- (6,5) -- cycle; % For hiding help lines behind psi
\filldraw[fill=blue,fill opacity=0.15,thick] (1,6) -- (1,2) -- (6,2) -- (6,4) -- (4,4) -- (4,6) -- cycle;
\node at (2.6,4) {$\satdom{\varphi^\tau}$};
% Polytope psi
\filldraw[fill=red,fill opacity=0.3,thick] (2,1) -- (8,1) -- (8,5) -- (6,5) -- cycle;
\node at (7,3) {$\satdom{\psi^\tau}$};

% Redraw phi boundary
\draw[thick] (1,6) -- (1,2) -- (6,2) -- (6,4) -- (4,4) -- (4,6) -- cycle;

\end{tikzpicture}

%% file: figs/eventually_example.tex
\begin{tikzpicture}[scale=0.55]
% First plot
% Axes
\draw[gray!50,very thin] (0,0) grid (8.8,8.8);
\draw[->] (0,0) -- (9,0);
\draw[->] (0,0) -- (0,9);
\draw (2pt,0) -- (-2pt,0) node[anchor=east,xshift=2pt] {\footnotesize $0$};
\draw (0,2pt) -- (0,-2pt) node[anchor=north,yshift=1pt] {\footnotesize $0$};

% Epsilon bounds
\draw[->] (4,0) -- (9,5);
\draw[->] (0,4) -- (5,9);

% Axis labels
\node[anchor=north] at (4,0) {$\varepsilon$};
\node[anchor=east] at (0,4) {$\varepsilon$};
\node[anchor=north] at (8.5,0) {$N+1$};
\node[anchor=north] at (-0.5,8.5) {$t_k$};

% Polytope
\fill[fill=white] (1,5) -- (1,0) -- (4,0) -- (7,3) -- (7,7) -- (5,7) -- (5,5) -- (4,5) -- (4,6) -- (2,6) -- (2,5) -- cycle;
\fill[fill=green!50!gray,opacity=0.3] (1,5) -- (1,0) -- (4,0) -- (7,3) -- (7,7) -- (5,7) -- (5,5) -- (4,5) -- (4,6) -- (2,6) -- (2,5) -- cycle;
\fill[pattern={Lines[angle=45,line width=2pt,distance=8pt,xshift=6pt]},pattern color=green!70!black,opacity=0.3] (1,5) -- (1,0) -- (4,0) -- (7,3) -- (7,7) -- (5,7) -- (5,5) -- (4,5) -- (4,6) -- (2,6) -- (2,5) -- cycle;
\draw[thick] (5,7) -- (5,5);
\draw[thick] (4,5) -- (4,6);
\draw[thick] (2,6) -- (2,5);
\draw[thick] (1,5) -- (1,0) -- (4,0) -- (7,3) -- (7,7);
\node at (3.7,3.3) {$\satdom{\eventually_{[1,2]}\varphi^\tau}$};
\draw[dashed,thick] (7,7) -- (5,7);
\draw[dashed,thick] (5,5) -- (4,5);
\draw[dashed,thick] (4,6) -- (2,6);
\draw[dashed,thick] (2,5) -- (1,5);

% Second plot
\begin{scope}[xshift=-11cm]

% Axes
\draw[gray!50,very thin] (0,0) grid (8.8,8.8);
\draw[->] (0,0) -- (9,0);
\draw[->] (0,0) -- (0,9);
\draw (2pt,0) -- (-2pt,0) node[anchor=east,xshift=2pt] {\footnotesize $0$};
\draw (0,2pt) -- (0,-2pt) node[anchor=north,yshift=1pt] {\footnotesize $0$};

% Epsilon bounds
\draw[->] (4,0) -- (9,5);
\draw[->] (0,4) -- (5,9);

% Axis labels
\node[anchor=north] at (4,0) {$\varepsilon$};
\node[anchor=east] at (0,4) {$\varepsilon$};
\node[anchor=north] at (8.5,0) {$N+1$};
\node[anchor=north] at (-0.5,8.5) {$t_k$};

% Polytope
\filldraw[fill=blue!15,thick] (3,5) -- (3,2) -- (5,2) -- (5,4) -- (6,4) -- (6,2) -- (8,4) -- (8,7) -- (7,7) -- (7,5) -- (5,5) -- (5,6) -- (4,6) -- (4,5) -- cycle;
\node at (4.1,4) {$\satdom{\varphi^\tau}$};

\end{scope}

\end{tikzpicture}

%% file: figs/until_example.tex
\begin{tikzpicture}[scale=0.95]

% Axes
% \draw[gray!50,very thin] (0,0) grid (11.8,8.8);
\draw[->] (0,0) -- (12,0);
\draw[->] (0,0) -- (0,9);
\draw (2pt,0) -- (-2pt,0) node[anchor=east,xshift=2pt] {\footnotesize $0$};
\draw (0,2pt) -- (0,-2pt) node[anchor=north,yshift=1pt] {\footnotesize $0$};

\draw (0,2) -- (-2pt,2) node[anchor=east,xshift=2pt] {\footnotesize $2$};
\draw (2,0) -- (2,-2pt) node[anchor=north,yshift=1pt] {\footnotesize $2$};

\draw (0,4) -- (-2pt,4) node[anchor=east,xshift=2pt] {\footnotesize $4$};
\draw (4,0) -- (4,-2pt) node[anchor=north,yshift=1pt] {\footnotesize $4$};

\draw (0,6) -- (-2pt,6) node[anchor=east,xshift=2pt] {\footnotesize $6$};
\draw (6,0) -- (6,-2pt) node[anchor=north,yshift=1pt] {\footnotesize $6$};

\draw (0,8) -- (-2pt,8) node[anchor=east,xshift=2pt] {\footnotesize $8$};
\draw (8,0) -- (8,-2pt) node[anchor=north,yshift=1pt] {\footnotesize $8$};

\draw (10,0) -- (10,-2pt) node[anchor=north,yshift=1pt] {\footnotesize $10$};

\node[anchor=north] at (11.7,0) {$N+1$};
\node[anchor=north] at (-0.5, 9) {$t_k$};

% Polytopes for phi
% \filldraw[fill=blue,fill opacity=0.15,thick] (0,1) -- (1,1) -- (1,2) -- (0,2) -- cycle;
\filldraw[fill=blue,fill opacity=0.15,thick] (3,4) -- (8,2) -- (8,6) -- (4,6) -- (4,7) -- (3,7) -- cycle;
% Polytope psi
\filldraw[fill=red,fill opacity=0.3,thick] (7,5) -- (9,3) -- (9,1) -- (11,1) -- (11,8) -- (8,8) -- cycle;

\fill[blue,opacity=0.07] (3,4) -- (3,2) -- (8,2);
\fill[red,opacity=0.1] (8,8) -- (7,8) -- (7,5) (7,5) -- (7,3) -- (9,3);

% \fill[green,opacity=0.2] (2.9,1.9) rectangle (6.1,6.1);
\fill[green,opacity=0.1] (3,2) rectangle (6,6);

% Since psi is now covering phi, make phi's overlapped borders black again
\draw[thick] (3,4) -- (8,2) -- (8,6) -- (4,6) -- (4,7) -- (3,7) -- cycle;

% Approximations
\draw[thick,dashed] (3,4) -- (3,2) -- (8,2); % Phi approx
\draw[thick,dashed] (8,8) -- (7,8) -- (7,3) -- (9,3); % Psi approx

% Partitioning for S1 and S2
\draw[dashed] (4,6) -- (3,6) (4,6) -- (4,2); % Partitioning for S1
\draw[dashed] (9,3) -- (9,8) (9,3) -- (11,3); % Partitioning for S2

% Refining S1
\draw[dotted,thick] (7,3) -- (7,2) (7,3) -- (3,3);

% Eventually offset
\fill[gray,opacity=0.15] (5.9,0) -- (5.9,-0.1) -- (6,0) -- cycle%
                        (1.9,5.9) -- (2,5.9) -- (2,6) -- cycle;
\filldraw[fill=gray,fill opacity=0.15] (2,0) rectangle (6,6);
\draw (1.9,-0.1) -- (2,0) (5.9,-0.1) -- (6,0) (5.9,5.9) -- (6,6) (1.9,5.9) -- (2,6);
\filldraw[fill=gray,fill opacity=0.2] (1.9,-0.1) rectangle (5.9,5.9);

% \draw (6.9,2.9) rectangle (8.1,6.1);
\fill[violet,opacity=0.1] (7,3) rectangle (8,6);
\fill[pattern={Lines[angle=45,line width=2pt,distance=6pt]},pattern color=violet,opacity=0.3] (7,3) rectangle (8,6);

% \draw (2.9,1.9) rectangle (6.1,6.1);
% \fill[pattern={Lines[angle=45,line width=2pt,distance=8pt]},pattern color=green!70!black,opacity=0.4] (3,2) rectangle (6,6);
\fill[pattern={Lines[angle=45,line width=2pt,distance=8pt,xshift=6pt]},pattern color=green!70!black,opacity=0.4] (3,2) rectangle (4,3);
\fill[pattern={Lines[angle=45,line width=2pt,distance=8pt]},pattern color=green!70!black,opacity=0.4] (4,2) rectangle (6,3);
\fill[pattern={Lines[angle=45,line width=2pt,distance=8pt]},pattern color=green!70!black,opacity=0.4] (3,3) rectangle (4,6);
\fill[pattern={Lines[angle=45,line width=2pt,distance=8pt,xshift=6pt]},pattern color=green!70!black,opacity=0.4] (4,3) rectangle (6,6);

\node[blue!60!black] at (4,7.4) {$\satdom{\varphi^a}$};
\node[red!35!black] at (11,8.4) {$\satdom{\psi^\tau}$};
\node[violet!70!black] (destlabel) at (8.5,2.4) {$S'$};
\node at (2.9,1) {$G$};
\node[green!35!black] (y) at (5.8,6.6) {$Y$};
\node[green!35!black] (s1) at (4.2,1.4) {$S_1$};

% \draw[->] (destlabel) to [bend left=30] (7.5,3.4);
\draw[->] (8.25,2.45) to [bend left=30] (7.5,3.4);
\draw[->] (y) to [bend left=20] (5,5.4);
% \draw[->] (s1) to [bend left=20] (3.5,2.3);
\draw[->] (3.95,1.55) to [bend left=20] (3.5,2.3);
\end{tikzpicture}

%% file: algo-until.tex
\begin{algorithm}[!t]
    \DontPrintSemicolon
    \SetNoFillComment
    \SetKwFor{ForEach}{for each}{\string:}{}
    \SetKwFunction{FDecompose}{decompose}
    \SetKwFunction{FPartition}{partition}
    \SetKwFunction{FRefine}{refine}
    \SetKwFunction{FOffset}{offset}
    \SetKwInOut{Input}{input}
    \SetKwInOut{Output}{output}
    \Input{$\satdom{\stlsubfmatom}$ = a non-convex polytope}
    \Input{$\satdom{\stlfm{\psi}^\tau}$ = a non-convex polytope}
    \Output{$Y \cap \satdom{\top}$, a non-convex polytope}
    $Y \defeq \emptyset$ \tcp*[l]{Initialize our results set to be empty.}
        \ForEach{polytope in $\satdom{\stlsubfmatom}$}{
            Find a bounding box of the polytope and store in set $\overline{A}$.\;
        }
        \ForEach{polytope in $\satdom{\stlfm{\psi}^\tau}$}{
            Find a bounding box of the polytope and store in set $\overline{B}$.\;
        }
        \tcc{Partition $\overline{A}$ and $\overline{B}$ into boxes (``blocks'') such that each side of a block is either a part of the boundary of $\overline{A}$, or is identical to the side of another block.}
        $\mathcal{S}_1 \defeq$ \FPartition{$\overline{A}$};\quad $\mathcal{S}_2 \defeq$ \FPartition{$\overline{B}$}\;
        \tcc{Produce a refined set with the same property above as $\mathcal{S}_1$.}
        $\mathcal{S}_3 \defeq \emptyset$ \tcp*[l]{Initialize the refined set.}
        \ForEach{bounding face $h$ of a block of $\mathcal{S}_2$}{
            \lIf{$h$ passes through $\mathcal{S}_1$}{separate $\mathcal{S}_1$ by $h$ and store the two pieces in $\mathcal{S}_3$.}
        }
        \ForEach{{\em destination block} $S' \in \mathcal{S}_3$ fully contained in   $\overline{B}$}{
            \tcp{Identify possible starting cuts.}
            $G \defeq S' \ominus \ilift{I}$\;
            $L \defeq \{S'\}$ \tcp*[l]{Initialize the set of blocks to track.}
            \While{$L$ is nonempty}{ 
                Pop a block $S$ from $L$.\;
                \lIf{$G \cap S$ is nonempty}{store $G \cap S$ in $Y$.}
                \ForEach{block $S'' \in \mathcal{S}_3$ sharing a side with a lower boundary of $S$}{
                    Push $S''$ to $L$.\;
                }
            }
        }
    \Return{$Y \cap \satdom{\top}$} 
    \caption{Algorithm \moveBackOut{}. (See  Fig.~\ref{fig:until}).}
    \label{algo:approx until satcuts}
\end{algorithm}

%% file: appendices.tex
\appendix

\section{Proofs}
\subsection{Theorem \ref{thm:valid sat sig}}
Before showing the proof for \autoref{thm:valid sat sig}, we need to show a couple of necessary lemmas. We begin by stating that for the non-timed fragment of \distl{} (i.e. formulas $\stlsubfmatom$), retimings do not affect satisfaction.

\begin{lemma}\label{lem:atom no retiming}
    Consider any retiming $\rho$ and any concut $C \in Gph(\rho)$. Then \\ $(\dstrsig_\rho, C[N+1]) \models \stlsubfmatom$ iff $\lambda_{\stlsubfmatom}(C) = \top$.
\end{lemma}

\begin{proof}
    Cases $\stlsubfmatom = \top$ and $\stlsubfmatom = \bot$ are trivial.
    Case $\stlsubfmatom = p_n$: $(\dstrsig_\rho, C[N+1]) \models p_n \iff  x_n(\rho_n^{-1}(C[N+1])) \geq p_n 
            \iff x_n(C[n]) \geq p_n 
            \iff  \lambda_{p_n}(C)$. Case $\stlsubfmatom = \lnot p_n$ is identical to $\stlsubfmatom = p_n$, replacing $\geq$ with $<$.
The remaining cases (conjunction and disjunction) are immediate.
\end{proof}

Next, the following lemma says that a retiming that passes through a satcut $C$ can always be adjusted to pass through any preceding cut $C'$ as well, without affecting satisfaction at $C$.
Another way of thinking about this is that cuts that precede $C$ are irrelevant for evaluating satisfaction at $C$, since the logic uses future tenses.
\begin{lemma}\label{lem:irr retiming cuts}
    For all cuts $C, C'$ where $C \preceq C'$, $\exists \rho \in P_C \cap P_{C'}: (\dstrsig_\rho, C'[N+1]) \models \stlsubfm$ iff $\exists \rho \in P_{C'}: (\dstrsig_\rho, C'[N+1]) \models \stlsubfm$.
\end{lemma}

\begin{proof}
    The $\implies$ direction is trivial, so we tackle the $\impliedby$ direction.
    The ordering $C \preceq C'$ has two cases: $C = C'$ and $C \prec C'$. When $C = C'$, or when $C=\mathbf{0}$, then the $\impliedby$ direction is also trivially true (recall that every retiming passes through 0 by definition), so for the rest of this proof we consider $\mathbf{0} \prec C \prec C'$.

    Since our formulation of \stl{} only has future-time operators, satisfaction of $(\dstrsig_\rho, C'[N+1]) \models \stlsubfm$ does not rely on retimed signal values prior to $C'[N+1]$. This means we can construct a retiming which passes through both $C'$ and $C$ and still allows for satisfaction of $\stlsubfm$.

    Consider the retiming $ \rho' \defeq (\rho'_1, \dots, \rho'_N)$ where for all $i \in [N]$:
    \[\rho'_i(t) \defeq \begin{cases}
        \frac{C[i]}{C[N+1]}t & t \leq C[N+1] \\
        \rho_i(t) & t \geq C'[N+1] \\
        \frac{C'[i] - C[i]}{C'[N+1] - C[N+1]}(t - C[N+1]) + C[i] & \text{otherwise} \\
    \end{cases}\]

    We next show that this agent retiming $\rho_i'$ is well-formed. To be well-formed as an agent retiming, we can show $\rho'_i$ is strictly increasing, continuous, $\rho_i'(0) = 0$, and $\forall t: |t - \rho_i'(t)| \leq \cskew$. This construction meets all three requirements. For clarity we show why this construction meets the third requirement. When $t \leq C[N+1]$, the maximum difference between $t$ and $\rho_i'(t)$ is at $t = C[N+1]$, where $\rho_i'(t) = C[i]$. Then $|t - \rho_i'(t)| \leq |C[i] - C[N+1]| \leq \cskew$, by the definition of a cut. When $C[N+1] < t \leq C'[N+1]$, the maximum difference between $t$ and $\rho_i'(t)$ is either at $t = C[N+1]$ where $\rho_i'(t) = C[i]$ or at $t = C'[N+1]$ where $\rho_i'(t) = C'[i]$. For either of these, the difference is less than or equal to $\cskew$, by the definition of a cut. Finally, when $t \geq C'[N+1]$, the slope is 1, meaning that the difference between $t$ and $\rho_i'(t)$ remains constant for the entire span, being $|C'[i] - C'[N+1]|$.

    This agent retiming $\rho_i'$ passes through both $C$ and $C'$. Since $x_i(\rho_i(t))$ only affects satisfaction when $t \geq C'[N+1]$, by construction $(\dstrsig_{\rho'}, C'[N+1]) \models \stlsubfm$ is equivalent to $(\dstrsig_\rho, C'[N+1]) \models \stlsubfm$. Therefore $\exists \rho \in P_{C'}: (\dstrsig_\rho, C'[N+1]) \models \stlsubfm$ implies $\exists \rho' \in P_C \cap P_{C'}: (\dstrsig_{\rho'}, C'[N+1]) \models \stlsubfm$.
\end{proof}

Now that we have shown these lemmas, we are ready for the proof of \autoref{thm:valid sat sig}.

\begin{proof}[\autoref{thm:valid sat sig}]
We consider each case of the recursive structure individually.
Cases $\stlsubfm = \top, \bot, p_n$ and $\lnot p_n$ are immediate or follow trivially from the proof of Lemma \ref{lem:atom no retiming}.
Case $\stlsubfm \land \stlfm{\psi}^a$:
\begin{flalign*}
    &\exists \rho \in P_C: (\dstrsig_\rho, C[N+1]) \models \stlsubfm \land \stlfm{\psi}^a \\
    \iff &\exists \rho \in P_C: ((\dstrsig_\rho, C[N+1]) \models \stlsubfm \text{ and } (\dstrsig_\rho, C[N+1]) \models \stlfm{\psi}^a) \\
    \iff &(\exists \rho \in P_C: (\dstrsig_\rho, C[N+1]) \models \stlsubfm) \land \lambda_{\stlfm{\psi}^a}(C) 
    \\
    & \quad \text{(By Lemma \ref{lem:atom no retiming} $\lambda_{\stlfm{\psi}^a}(C)$ does not depend on retiming)} \\
    \iff &\lambda_{\stlsubfm}(C) \land \lambda_{\stlfm{\psi}^a}(C) 
    \iff \lambda_{\stlsubfm \land \stlfm{\psi}^a}(C)
\end{flalign*}
The remaining cases of conjunction and disjunction follow a similar pattern.

Case $\eventually_I \stlsubfm$:
$\exists \rho \in P_C: (\dstrsig_\rho, C[N+1]) \models \eventually_I \stlsubfm 
\iff \exists \rho \in P_C: \exists t' \in C[N+1] \oplus I: (\dstrsig_\rho, t') \models \stlsubfm$
    We can construct a $C' \in Gph(\rho)$ such that every instance of $t'$ is replaced by $C'[N+1]$. Since the only index of $C'$ being used is $C'[N+1]$, this conversion can be treated as an iff.
    \begin{flalign*}
        \iff &\exists \rho \in P_C: \exists C' \in Gph(\rho) \text{ with } \\
        & (C' \succeq C \text{ and } C'[N+1] \in C[N+1] \oplus I):  (\dstrsig_\rho, C'[N+1]) \models \stlsubfm \\
        \iff &\exists C' \succeq C \text{ with } C'[N+1] \in C[N+1] \oplus I: \\
        &\qquad \exists \rho \in P_C \cap P_{C'}: (\dstrsig_\rho, C'[N+1]) \models \stlsubfm \\
        \iff &\exists C' \succeq C \text{ with } C'[N+1] \in C[N+1] \oplus I: \\
        &\qquad \exists \rho \in P_{C'}: (\dstrsig_\rho, C'[N+1]) \models \stlsubfm \qquad\enspace \text{(\autoref{lem:irr retiming cuts})} \\
        \iff &\exists C' \succeq C \text{ with } C'[N+1] \in C[N+1] \oplus I: \lambda_{\stlsubfm}(C') \\
        \iff &\lambda_{\eventually_I \stlsubfm}(C)
    \end{flalign*}

Case $\always_I \stlsubfmatom$ follows a similar pattern.

Case $\stlsubfmatom \until{I} \stlfm{\psi}^\tau$:
    \begin{flalign*}
        &\exists \rho \in P_C: (\dstrsig_\rho, C[N+1]) \models \stlsubfmatom \until{I} \stlfm{\psi}^\tau &\\
        \iff &\exists \rho \in P_C: \exists t' \in C[N+1] \oplus I: (\dstrsig_\rho, t') \models \stlfm{\psi}^\tau \text{ and } \\
        &\qquad \forall t'' \in [C[N+1], t']: (\dstrsig_\rho, t'') \models \stlsubfmatom
    \end{flalign*}
    We can construct a $C' \in Gph(\rho)$ such that every instance of $t'$ is replaced by $C'[N+1]$. Since the only index of $C'$ being used is $C'[N+1]$, this conversion can be treated as an iff.
    \begin{flalign*}
        &\exists \rho \in P_C: \exists C' \in Gph(\rho) \text{ with } (C' \succeq C \text{ and } C'[N+1] \in C[N+1] \oplus I): \\
        &\qquad (\dstrsig_\rho, C'[N+1]) \models \stlfm{\psi}^\tau \text{ and } &\\
        & \qquad \forall t'' \in [C[N+1], C'[N+1]]: (\dstrsig_\rho, t'') \models \stlsubfmatom \\
        \text{iff } &\exists C' \succeq C \text{ with } C'[N+1] \in C[N+1] \oplus I: \exists \rho \in P_C \cap P_{C'}: \\
        &\qquad (\dstrsig_\rho, C'[N+1]) \models \stlfm{\psi}^\tau \text{ and } \\
        & \qquad \forall t'' \in [C[N+1], C'[N+1]]: (\dstrsig_\rho, t'') \models \stlsubfmatom
    \end{flalign*}
    Here we construct a $C'' \in Gph(\rho)$ such that every instance of $t''$ is replaced by $C''[N+1]$. Furthermore, rather than identifying $C''[N+1] \in [C[N+1], C'[N+1]]$, we can look at $C'' \in [C, C']$ since for a retiming $\rho$ there exists exactly one cut in $\rho$ at index $N+1$.
    \begin{flalign*}
        \iff &\exists C' \succeq C \text{ with } C'[N+1] \in C[N+1] \oplus I: \exists \rho \in P_C \cap P_{C'}: \\
        &\qquad (\dstrsig_\rho, C'[N+1]) \models \stlfm{\psi}^\tau \text{ and } &\\
        & \qquad \forall C'' \in Gph(\rho) \text{ with } C'' \in [C, C']: (\dstrsig_\rho, C''[N+1]) \models \stlsubfmatom \\
        \implies &\exists C' \succeq C \text{ with } C'[N+1] \in C[N+1] \oplus I: \\
        &\qquad (\exists \rho \in P_C \cap P_{C'}: (\dstrsig_\rho, C'[N+1]) \models \stlfm{\psi}^\tau) \text{ and } \\
        & \qquad (\exists \rho' \in P_C \cap P_{C'}: \forall C'' \in Gph(\rho') \text{ with } C'' \in [C, C']: \\
        &\qquad (\dstrsig_{\rho'}, C''[N+1]) \models \stlsubfmatom)
    \end{flalign*}
    Notice that this is a single-directional implication, given by a weakening property of first-order logic. Next we show that the implication can go the opposite direction as well.

    Consider a retiming $\rho'' \defeq (\rho''_1, \dots, \rho''_N)$ where for all $i \in [N]$:
    \[\rho''_i(t) \defeq \begin{cases}
        \rho'[i](t) & t < C'[N+1] \\
        \rho[i](t) & t \geq C'[N+1]
    \end{cases}\]
    Since both $\rho$ and $\rho'$ pass through $C'$, this is a valid retiming. This retiming can replace both $\rho$ and $\rho'$, allowing the satisfaction properties to remain true.

    Consider $(\dstrsig_\rho, C'[N+1]) \models \stlfm{\psi}^\tau$: Since our formulation of \stl{} only considers future-time operators, satisfaction only relies on the retiming function for moments $> C'[N+1]$ (\autoref{lem:irr retiming cuts}). This means that $(\dstrsig_\rho, C'[N+1]) \models \stlfm{\psi}^\tau \implies (\dstrsig_{\rho''}, C'[N+1]) \models \stlfm{\psi}^\tau$.

    Next, consider $(\dstrsig_{\rho'}, C''[N+1]) \models \stlsubfmatom$: $\stlsubfmatom$ ensures that satisfaction only relies on $C''$. Since the maximum value of $C''$ is $C'$, this means that $(\dstrsig_{\rho'}, C''[N+1]) \models \stlsubfmatom \implies (\dstrsig_{\rho''}, C''[N+1]) \models \stlsubfmatom$.

    Combining, these facts allow us to consider existence of a single retiming, meaning that we've shown implication in the opposite direction ($\impliedby$).
    \begin{eqnarray*}
        \iff &\exists C' \succeq C \text{ with } C'[N+1] \in C[N+1] \oplus I: \\
        & (\exists \rho \in P_{C'}: (\dstrsig_\rho, C'[N+1]) \models \stlfm{\psi}^\tau) \text{ and } \text{(\autoref{lem:irr retiming cuts})} & \\
        &  (\exists \rho \in P_C \cap P_{C'}: \forall C'' \in Gph(\rho) \text{ with } \\
        &\qquad C'' \in [C, C']: (\dstrsig_\rho, C''[N+1]) \models \stlsubfmatom) \\
        \iff &\exists C' \succeq C \text{ with } C'[N+1] \in C[N+1] \oplus I: \lambda_{\stlfm{\psi}^\tau}(C') \text{ and } \\
        &  \exists \rho \in P_C \cap P_{C'}: \forall C'' \in Gph(\rho) \text{ with } \\
        & C'' \in [C, C']: \lambda_{\stlsubfmatom}(C'')  \text{ (\autoref{lem:atom no retiming})} & \\
        \iff &\lambda_{\stlsubfmatom \until{I} \stlfm{\psi}^\tau}(C)
    \end{eqnarray*}
\end{proof}

\subsection{Theorem \ref{thm:satdomain equalities}}
    The base cases $\satdom{\top}$, $\satdom{\bot}$, $\satdom{p_n}$, and $\satdom{\lnot p_n}$ are all simple.
    
    Case $\satdom{\stlsubfmatom \lor \stlfm{\psi}^a}$:
            $C \in \satdom{\stlsubfmatom \lor \stlfm{\psi}^a}$ 
            iff  $\lambda_{\stlsubfmatom \lor \stlfm{\psi}^a}(C)$ (by \autoref{thm:valid sat sig})
            iff $\lambda_{\stlsubfmatom}(C) \lor \lambda_{\stlfm{\psi}^a}(C)$ 
            iff $(C \in \satdom{\stlsubfmatom})$ or $(C \in \satdom{\stlfm{\psi}^a})$ (by the induction hypothesis) 
            iff  $C \in \satdom{\stlsubfmatom} \cup \satdom{\stlfm{\psi}^a}$.
    
    Cases $\satdom{\stlsubfm \lor \stlfm{\psi}^\tau}$, $\satdom{\stlsubfmatom \land \stlfm{\psi}^a}$ and $\satdom{\stlsubfm \land \stlfm{\psi}^a}$ follow the same proof as the previous case, with the obvious changes.

Case $\satdom{\eventually_I \stlsubfm}$: $C \in \satdom{\eventually_I \stlsubfm}$ iff $\lambda_{\eventually_I \stlsubfm}(C)$, which is equivalent to
\begin{equation*}
\exists C' \succeq C \text{ with } C'[N+1] \in C[N+1] \oplus I: \lambda_{\stlsubfm}(C')
\end{equation*}
which is equivalent to 
\[\exists C' \text{ s.t. }  \lambda_{\stlsubfm}(C') \textrm{ and } C' \in C \oplus \ilift{I}\]
The transition between the last 2 lines is justified as follows: if $I=[0,a]$ then either $C'[N+1]=0$ and $C'=C$ or $C'[N+1]>0$ and $C' \succ C$. In the former case, $C' = C+\mathbf{0}$, and in the latter $C'=C+\mathbf{v}$ for some strictly positive vector $\mathbf{v}$. So $C'\in C \oplus \ilift{I}$ is equivalent to $C'\succeq C$ and $C'[N+1] \in C[N+1] \oplus I$.
If $I$ does not contain 0, then only the second case is possible. 

Continuing with the list of equivalencies, 
            $\iff  \exists C' \in \satdom{\stlsubfm}: C' \in C \oplus \ilift{I} 
            \iff  \exists C' \in \satdom{\stlsubfm}: C' \in \{C + K \mid K \in \ilift{I}\} 
            \iff  \exists C' \in \satdom{\stlsubfm}: \exists K \in \ilift{I}: C' = C + K 
            \iff  \exists C' \in \satdom{\stlsubfm}: \exists K \in \ilift{I}: C = C' - K 
            \iff  C \in \{C' - K \mid C' \in \satdom{\stlsubfm}, K \in \ilift{I}\} 
            \iff  C \in (\satdom{\stlsubfm} \ominus \ilift{I})$

    Case $\satdom{\always_I \stlsubfmatom}$:
        We show 
        that \textbf{(A)} $\lambda_{\always_{[b,c]} \stlsubfmatom}(C) = \lambda_{\eventually_{[b,b]} \always_{[0,c-b]} \stlsubfmatom}(C)$, 
        then show that \textbf{(B) }$\lambda_{\always_{[0,c-b]} \stlsubfmatom}(C) = \lambda_{\stlsubfmatom \until{[c-b,c-b]} \top}(C)$. 
        Together, these show that $\lambda_{\always_{[b,c]} \stlsubfmatom}(C) = \lambda_{\eventually_{[b,b]}\left(\stlsubfmatom \until{[c-b,c-b]} \top\right)}(C)$, and this gives us the desired result since the satdomain is the support of the satisfaction signal $\lambda$.

        First showing \textbf{(A)}. Indeed we have the following equivalences: $\lambda_{\always_{[b,c]} \stlsubfmatom}(C) = \top$ iff 
        \begin{flalign*}
            & \exists \rho \in P_C: \forall C' \in Gph(\rho) \text{ with } 
            C'[N+1] \in C[N+1] \oplus [b,c]: \\
            &\qquad \lambda_{\stlsubfmatom}(C') 
            \\
            \iff & \exists \rho \in P_C: (\forall C' \in Gph(\rho) \text{ with } 
            C'[N+1] \in C[N+1] \oplus [b,c]: \\
            &\qquad \lambda_{\stlsubfmatom}(C')) 
            \\
            & \land \exists C'' \in Gph(\rho): C''[N+1] = C[N+1] + b \land C'' \succeq C 
            \\
            \iff & \exists \rho \in P: \exists C'' \in Gph(\rho): C''[N+1] = C[N+1] + b \land C'' \succeq C 
            \\
            & \land \forall C' \in Gph(\rho) \text{ with } C'[N+1] \in C[N+1] \oplus [b,c]: \lambda_{\stlsubfmatom}(C') 
             \\
            \iff & \exists C'': C''[N+1] = C[N+1] + b \land C'' \succeq C 
            \\
            & \land \exists \rho \in P: C'' \in Gph(\rho) \land \forall C' \in Gph(\rho) \text{ with } \\
            &\qquad C'[N+1] \in C[N+1] \oplus [b,c]: \lambda_{\stlsubfmatom}(C') 
            \\
            \iff & \exists C'' \succeq C \text{ with } C''[N+1] = C[N+1] + b: 
            \\
            & \exists \rho \in P_{C''} \text{ s.t. } 
            \forall C' \in Gph(\rho) \text{ with } \\
            &\qquad C'[N+1] \in C[N+1] \oplus [b,c]: \lambda_{\stlsubfmatom}(C') 
            \\
            \iff & \exists C'' \succeq C \text{ with } C''[N+1] \in C[N+1] \oplus [b,b]: 
            \\
            & \exists \rho \in P_{C''} \text{ s.t. } 
            \forall C' \in Gph(\rho) \text{ with } \\
            &\qquad C'[N+1] \in C''[N+1] \oplus [0,c-b]: \lambda_{\stlsubfmatom}(C') 
            \\
            \iff & \lambda_{\eventually_{[b,b]}\always_{[0,c-b]} \stlsubfmatom}(C)=\top 
        \end{flalign*}

        Next showing \textbf{(B)}. Indeed $\lambda_{\always_{[0,c-b]} \stlsubfmatom}(C)=\top$ iff $\exists \rho \in P_C$ s.t.: 
        \begin{flalign*}
            & \forall C' \in Gph(\rho) \text{ with } 
            C'[N+1] \in C[N+1] \oplus [0,c-b]: \lambda_{\stlsubfmatom}(C') 
            \\
            \iff & \exists C'' \in Gph(\rho): C''[N+1] = C[N+1] + c - b 
            \\
             & \land \forall C' \in Gph(\rho) \text{ with } C'[N+1] \in C[N+1] \oplus [0,c-b]: \lambda_{\stlsubfmatom}(C') \\
            \iff & \exists C'' \in Gph(\rho): C''[N+1] = C[N+1] + c - b 
            \\
            & \land \forall C' \in Gph(\rho) \text{ with } C'[N+1] \in [C[N+1],C''[N+1]]: \\ 
            &\lambda_{\stlsubfmatom}(C') 
        \end{flalign*}
Equivalently $\exists C'' \succeq C \text{ with } C''[N+1] \in C[N+1] \oplus [c-b,c-b]$ and $\exists \rho \in P_C \cap P_{C''}: \forall C' \in Gph(\rho) \cap [C,C'']: \lambda_{\stlsubfmatom}(C')$. In other words, $\lambda_{\stlsubfmatom \until{[c-b,c-b]} \top}(C)$. 

\subsection{Lemma \ref{thm:satdomains are NCPs}}
\begin{proof}
    The cases $\satdom{\bot}$, $\satdom{\top}$ case are trivial.
    $\satdom{p_n}$ is the intersection of polytope $\satdom{\top}$ with a union of polytopes, the latter being the set of concuts $C$ s.t. $x_n(C[n])\geq p_n$. Indeed, given the constraints on our signals in \autoref{def:dstr signal}, $\{t\such x_n(t)\geq p_n\}$ is a disjoint union of intervals. 
    Therefore $\{C \in \eccut \such x_n(C[n])\geq p_n\} = \nnreals^{n-1}\times \{t\such |t-C[N+1]|\leq \varepsilon \text{ and } x_n(t)\geq p_n\} \times \nnreals^{N+1-n}$ is a union of polytopes.
    Similarly for $\satdom{\lnot p_n}$.
    
    The $\lor$ and $\land$ cases follow because the union/intersection of unions of non-convex polytopes (UNCPs) is a UNCP.
    $\satdom{\eventually_I \stlsubfm}$ is a UNCP because $\ominus$ preserves convexity -- if $\satdom{\stlsubfm}$ is a UNCP (the same is true of $\ilift{I}$), it is also a union of polytopes, so operating $\ominus$ on pairs of polytopes and then unioning produces a UNCP.
    $\satdom{\always_{[b,c]} \stlsubfmatom}$ is a UNCP because it is equivalent to a combination of other cases we've shown.

    All operations in \autoref{thm:satdomain captures satisfaction} (namely, unions, intersections, and shifts) produce non-convex polytopes from NCP inputs. \moveBackOut{} thus takes in non-convex polytopes (NCPs), creates box outer-approximations of them, partitions the boxes, then takes intersections of the parts, thus again producing NCPs.
\end{proof}

\subsection{Lemma \ref{lem:moveback approx}}
We show that $\text{\moveBackOut}(\stlsubfmatom \until{I} \stlfm{\psi}^\tau)$ is a superset of $\satdom{\stlsubfmatom \until{I} \stlfm{\psi}^\tau}$. The inner-approximation case follows similar reasoning.

The inputs to \moveBackOut{} are non-convex polytopes (NCPs) $A_1$ and $B_1$ which are supersets of $\satdom{\stlsubfmatom}$ and $\satdom{\stlfm{\psi}^\tau}$, respectively. 
We show that every satcut is contained in the algorithm's output $Y \cap \satdom{\top}$.

Without loss of generality, it is enough to prove the result for the case of a single block $S' \subseteq \overline{B}$ (as illustrated in \autoref{fig:until}).

We say that a cut $C'$ is $I$-away from a cut $C$ when $C'[N+1] \in C[N+1] \oplus I$ and $C \preceq C'$. Then we claim that $G$ contains all points that are $I$-away from some point in the destination $S'$.
To see this, consider that by the definition of $\ilift{I}$, for every point $C' \in S'$ we have that $S' \ominus \ilift{I}$ contains every point $C \preceq C'$ that is $I$-away from $C'$. This means that every satcut $C$ is in $G$.

Next we claim that for all blocks $S_1, \dots, S_k$ from $L$ (for some natural number $k$), $G \cap (S_1 \cup \dots \cup S_k)$ contains every point that 1) is $I$-away from $S'$ and 2) has a strictly increasing path to it within $\overline{A_1}$.
By the definition of a lower boundary, every strictly increasing path within $\overline{A_1}$ crossing into a block $S$ must pass through a lower boundary of $S$. Since every block $S''$ in $\overline{A_1}$ sharing a lower boundary of $S$ is added to $L$, then every strictly increasing path within $\overline{A_1}$ which ends in $S'$ passes through a sequence of blocks where each block is included in $L$ at some point. Thus the union of all blocks $S_1, \dots, S_k$ from $L$ contains all strictly increasing paths to $S'$ within $\overline{A_1}$. Since we have shown that $G$ contains all points that are $I$-away from some point in $S'$, then the intersection $G \cap (S_1 \cup \dots \cup S_k)$ contains every point that 1) is $I$-away from $S'$ and 2) has a strictly increasing path to it within $\overline{A_1}$.

By \autoref{def:sat sig}, this means that every satcut $C$ is in some $G \cap S$. Since every point in $G \cap S$ is added to the results of the algorithm (and all satcuts are contained in $\satdom{\top}$ by definition), every satcut $C$ is returned by the algorithm. Thus we conclude the proof.

\subsection{Complexity Analysis  Eq.~\eqref{eq:complexity bound}}
\label{sec:appendix complexity proof}
We only show the source of the fastest growing term, $M_1^{\lfloor \frac{N+1}{2}\rfloor}$.
 Finding a bounding box for a polytope in $\reals^N$ with $V$ vertices costs $O(N \cdot V)$, since one has to enumerate the vertices and find their minimum and maximum along each dimension. 
To retrieve the worst-case number of vertices, we consider the dual of a cyclic polytope. By the upper bound theorem \cite{seidel1995upper}, cyclic polytopes have the largest possible number of faces for a given number of vertices, so the dual provides the largest possible number of vertices for a given number of faces. This number is $O\left(M^{\lfloor \frac{N+1}{2}\rfloor}\right)$ (see Theorem 5.4.5 of \cite{matousek2013lectures}). Thus the full complexity of finding a bounding box is $O\left(N \cdot M^{\lfloor \frac{N+1}{2}\rfloor}\right)$.